\documentclass{llncs}[10pt]
\usepackage{graphicx}
\usepackage{latexsym,amssymb,amsmath,color}
\usepackage[letterpaper]{geometry}
\usepackage[algoruled,linesnumbered]{algorithm2e}

\begin{document}

\title{Approximation algorithm for the minimum directed tree cover} 
\author{Viet Hung Nguyen}

\institute{LIP6, Universit\'e Pierre et Marie Curie Paris 6, 4 place Jussieu, Paris, France}

\date{Received: date / Accepted: date}

\maketitle

\begin{abstract}
Given a directed graph $G$ with non negative cost on the arcs, a directed tree cover of $G$ is a rooted directed tree such that 
either head or tail (or both of them) of every arc in $G$ is touched by $T$. The minimum directed tree 
cover  problem (DTCP) is to find a directed tree cover of minimum cost. The problem is known to be $NP$-hard. In this paper,  we show that the weighted  Set Cover Problem (SCP) 
is a special case of DTCP. Hence, one can expect at best to approximate DTCP with the same ratio as for SCP.
We show that this expectation can be satisfied in some way by designing a purely combinatorial  approximation algorithm for the DTCP and proving that the approximation 
ratio of the algorithm is $\max\{2, \ln(D^+)\}$  with $D^+$ is the maximum outgoing degree of the nodes in $G$. 
\end{abstract}
\section{Introduction}
\label{}
Let $G=(V,A)$ be a directed graph with a (non negative) cost function 
$c: A \Rightarrow \mathbb{Q}_{+}$ defined on the arcs. Let $c(u,v)$ denote the cost 
of the arc $(u,v) \in A$.  
A \textit{directed tree} cover is a weakly connected subgraph $T=(U,F)$ such that
\begin{enumerate}
\item for every $e \in A$,  $F$ contains an arc $f$ intersecting $e$, i.e. 
 $f$ and $e$ have an end-node in common.
\item $T$ is a rooted branching.  
\end{enumerate}
The \textit{minimum directed tree 
cover problem} (DTCP) is to find a directed tree  cover of minimum cost.
Several related problems to DTCP have been investigated, in particular:
\begin{itemize}
 \item  its undirected counterpart, the minimum tree cover problem (TCP) and
\item the tour cover problem in which $T$ is a tour (not necessarily simple) instead of a tree.
This problem has also two versions: undirected (ToCP) and directed (DToCP).
\end{itemize}
We discuss first about TCP which has been intensively studied in recent years.
The TCP is introduced in 
a paper by Arkin et al. \cite{Arkin} where they were motivated by a problem of locating tree-shaped facilities on a graph such that 
all the nodes are dominated by chosen facilities. They proved the $NP$-hardness of TCP by observing that 
the unweighted case of TCP is equivalent to the \textit{connected vertex cover} problem, which in fact is known to be as hard (to approximate) 
as the vertex cover problem \cite{Garey}.
 Consequently, DTCP is also $NP$-hard since the TCP can be easily transformed to an instance of DTCP by replacing every edge 
by the two arcs of opposite direction between the two end-nodes of the edge.
In their paper, Arkin et al. presented a 2-approximation algorithm for the unweighted case of TCP, as well as 3.5-approximation algorithm for general costs. 
Later, Konemann et al. \cite{Koneman} and Fujito \cite{Fujito1} independently designed a 3-approximation algorithm
for TCP using a bidirected formulation. 
They solved a linear program (of exponential size) to find a vertex cover $U$ and then 
they found a Steiner tree with $U$ as the set of terminals. 
Recently, Fujito \cite{Fujito} and Nguyen \cite{Nguyen1} 
propose separately two different approximation algorithms achieving 2 the currently best approximation ratio.
Actually, the algorithm in \cite{Nguyen1} is expressed for the TCP when costs satisfy the triangle inequality but one can suppose this 
for the general case without loss generality. The algorithm in \cite{Fujito} is very interesting in term of complexity since it is 
a primal-dual based algorithm and thus purely combinatorial. In the prospective section of \cite{Koneman} and \cite{Fujito}, the authors presented  DTCP as a wide open problem for further research on the topic. 
In particular, Fujito  \cite{Fujito} pointed out that his approach for TCP can be extended to give 
a 2-approximation algorithm for the unweighted case of DTCP but falls short once arbitrary costs are allowed.  \\
For ToCP,  a 3-approximation algorithm has been developed in \cite{Koneman}. The principle of this algorithm is 
similar as for TCP, i.e. it solved a linear program (of exponential size) to find a vertex cover $U$ and then 
found a traveling salesman tour over the subgraph induced by $U$.
Recently, Nguyen \cite{Nguyen2} considered  DToCP and extended the approach in \cite{Koneman} to obtain a $2\log_2(n)$-approximation algorithm for DToCP. We can similarly adapt the method in 
\cite{Koneman} for TCP to DTCP but we will have to find a directed Steiner tree 
with $U$ a vertex cover as the terminal set. Using the best known approximation algorithm by Charikar et al. \cite{Charikar} for the minimum Steiner directed tree problem, we obtain a ratio of $(1 +\sqrt{|U|}^{2/3}log^{1/3}(|U|))$ for DTCP which is worse than a logarithmic ratio.\\
In this paper, we improve this ratio by giving a logarithmic ratio approximation algorithm for DTCP. In particular, we
show that the weighted Set Cover Problem (SCP) is a special case of DTCP and the transformation is approximation preserving. Based on the known complexity results for SCP, we can only expect a logarithmic ratio for the approximation of DTCP. Let $D^+$ be the maximum outgoing degree of the nodes in $G$, we design a primal-dual $\max\{2, \ln(D^+)\}$-approximation algorithm for DTCP which is thus somewhat best possible.\\ 
The paper is organized as follows. In the remaining of this section, we will define the notations that will be used in the papers. In Section 2, we present an integer formulation  and  state a primal-dual algorithm for DTCP. Finally, we prove the validity of the algorithm and its approximation ratio. \\ 
Let us introduce the notations that will be used in the paper. 
Let $G=(V,A)$ be a digraph with vertex set $V$ and arc set $A$. 
Let $n=|V|$ and $m=|A|$.
If $x \in \mathbb{Q}^{|A|}$ is a vector indexed by the 
arc set $A$ and $F \subseteq E$ is a subset of arcs, we use $x(F)$ to 
denote the sum of values of $x$ on the arcs in $F$, $x(F)=\sum_{e \in F}x_e$.
Similarly, for a vector $y \in \mathbb{Q}^{|V|}$ indexed by the nodes and $S \subseteq V$ is a subset of nodes, let $y(S)$ denote  the sum of 
values of $y$ on the nodes in the set $S$. 
For a subset of nodes $S \subseteq V$, let $A(S)$ denote the set of the arcs having both end-nodes in $S$. Let $\delta^+(S)$(respectively $\delta^-(S)$) denote the set of the arcs having only the tail (respectively head) in $S$. 
We will call $\delta^+(S)$ the \textit{outgoing cut} associated to $S$, $\delta^-(S)$ the \textit{ingoing cut} associated to $S$.
For two subset $U, W \subset V$ such that $U \cap W =\emptyset$, let $(U:W)$ be the set of the arcs having the tail in $U$ and 
the head in $W$.  For $u \in V$, we say $v$ an \textit{outneighbor} (respectively \textit{inneighbor}) of $u$ if $(u,v) \in A$ (respectively 
$(v,u) \in A$). For the sake of simplicity, in clear contexts, the singleton $\{u\}$ will be denoted simply by $u$.\\
For an arc subset $F$ of $A$, let $V(F)$ denote the set of end-nodes of all the arcs in $F$. 
We say $F$ \textit{covers} a vertex subset $S$ if $F \cap \delta^-(S) \neq \emptyset$. We say $F$ is a cover for the graph $G$ if 
for all arc $(u,v) \in A$, we have $F \cap \delta^-(\{u,v\}) \neq \emptyset$.\\
When we work on more than one graph, we specify the graph in the index of the notation, e.g. $\delta^+_{G}(S)$ will denote $\delta^+(S)$ 
in the graph $G$. By default, the notations without indication of the the graph in the index are applied on $G$.
\section{Minimum $r$-branching cover problem}
Suppose that $T$ is a directed tree cover of $G$ rooted in $r \in V$, i.e. $T$ is a branching, $V(T)$ is a vertex cover in $G$ and there is a directed path in $T$ from $r$ to any other node in $V(T)$. In this case, we call $T$, a 
\textit{$r$-branching cover}. Thus, DTCP can be divided into $n$ subproblems in which we find a minimum $r$-branching cover for all $r \in V$.
By this observation, in this paper,  we will focus on approximating the minimum $r$-branching cover for a specific vertex $r \in V$.
An approximation algorithm for DTCP is then simply resulted from applying $n$ times the algorithm  for the minimum $r$-branching cover for each 
$r \in V$.
\subsection{Weighted set cover problem as a special case}
Let us consider any instance $\mathcal{A}$ of the weighted Set Cover Problem (SCP) with a set $E=\{e_1,e_2,\ldots,e_p\}$ of ground elements, 
and a collection of subsets $S_1$, $S_2$, \ldots, $S_q \subseteq E$ with corresponding non-negative weights $w_1$, $w_2$, \ldots, $w_q$.
The objective is to find a set $I \subseteq\{1,2,\ldots,q\}$ that minimizes $\displaystyle \sum_{i \in I}w_i$, such that $\displaystyle \bigcup_{i \in I}S_i=E$.
We transform this instance to an instance of the minimum $r$-branching cover problem in some graph $G_1$ as follows. We create a node $r$, $q$ nodes $S_1$, $S_2$, \ldots, $S_q$ and
 $q$ arcs $(r,S_i)$ with weight $w_i$. We then add $2p$ new nodes $e_1$, \ldots, $e_p$ and $e'_1$, \ldots, $e'_p$. If $e_k \in S_i$ for some $1 \leq k \leq p$ and 
$1 \leq i \leq q$, we create an arc $(S_i,e_k)$ with weight 0 (or a very insignificant positive weight). At last, we add an arc $(e_k,e'_k)$ of weight 0 (or a very insignificant positive weight) for each $1 \leq k \leq p$. 
\begin{lemma}
 Any $r$-branching cover in $G_1$ correspond to a set cover in $\mathcal{A}$ of the same weight and vice versa.
\end{lemma}
\begin{proof}
 Let us consider any $r$-branching cover $T$ in $G_1$. Since $T$ should cover all the arcs $(e_k,e'_k)$ for $1 \leq k \leq n$, $T$ contains 
the nodes $e_k$. By the construction of $G_1$, these nodes are connected to $r$ uniquely through the nodes $S_1$, \ldots $S_q$ with the corresponding cost
$w_1$, \ldots, $w_q$. Clearly, the nodes $S_i$ in $T$ constitute a set cover in $\mathcal{A}$ of the same weight as $T$.  It is then easy to see that 
any set cover in $\mathcal{A}$ correspond to $r$-branching cover in $G_1$ of the same weight.
\end{proof}
Let $D^+_r$ be the maximum outgoing degree of the nodes (except $r$) in $G_1$. We can see that $D^+_r=p$, the number of 
ground elements in $\mathcal{A}$. 
Hence, we have  
\begin{corollary}
 Any $f(D^+_r)$-approximation algorithm for the minimum $r$-branching cover problem is also an $f(p)$-approximation algorithm for SCP where $f$ is a function from $\mathbb{N}$ to $\mathbb{R}$.
\end{corollary}
Note that the converse is not true.
As a corollary of this corollary, we have the same complexity results for the minimum $r$-branching cover problem as known results for SCP \cite{Lund,Feige,Raz,Arora}. Precisely,
\begin{corollary} \hspace{0.5cm}\\
\begin{itemize}
 \item If there exists a $c\ln(D^+_r)$-approximation algorithm for the minimum $r$-branching cover problem where $c <1$ then $NP \subseteq DTIME(n^{\{O(\ln^k(D^+_r))\}})$.
\item There exists some $0<c<1$ such that if there exists a $c\log(D^+_r)$-approximation algorithm for the minimum $r$-branching cover problem, then $P=NP$.
\end{itemize}
\end{corollary}
Note that this result does not contradict the Fujito's result about an approximation ratio 2 for the unweighted DTCP because in our transformation 
we use arcs of weight 0 (or a very insignificant fractional positive weight) which are not involved in an instance of unweighted DTCP.\\ 
Hence in some sense, the $\max\{2,\ln(D^+_r)\}$ approximation algorithm that we are going to describe in the next sections seems to be best possible for the general weighted DTCP.
\section{Integer programming formulation for minimum $r$-branching cover}\label{ip}
We use  a formulation inspired from the one in \cite{Koneman} designed originally for the TCP.  The formulation  is as follows:
for a fixed root $r$, define $\mathcal{F}$ to be the set of all subsets $S$ of $V \setminus \{r\}$ such that $S$ induces 
at least one arc of $A$,
$$\mathcal{F} = \{S \subseteq V \setminus \{r\} ~ | ~ A(S) \neq \emptyset \}. $$
Let $T$ be the arc set of a directed tree cover of $G$ containing $r$, $T$ is thus a branching rooted at $r$. Now for every $S \in \mathcal{F}$, 
at least one node, saying $v$, in $S$ should belong to $V(T)$. By definition of directed tree cover there is a path from $r$ to $v$ in $T$ and as 
$r \notin S$, this path should contain at least one arc in $\delta^-(S)$.  This allows us to derive the following \textit{cut} constraint which is valid for the DTCP:
$$\sum_{e \in \delta^-(S)}x_e \geq 1 \mbox{ for all $S \in \mathcal{F}$}$$
This leads to  the following IP formulation for the minimum $r$-branching cover.
\[ \min \sum_{e \in A}c(e)x_e \]
$$\sum_{e \in \delta^-(S)}x_e \geq 1 \mbox{ for all $S \in \mathcal{F}$}$$
$$ x \in \{0,1\}^{A}.$$
A trivial case for which this formulation has no constraint is when $G$ is a $r$-rooted star but in this case the optimal solution is trivially 
the central node $r$ with cost 0.\\ 
Replacing the integrity constraints by
$$ x \geq 0, $$
we obtain the linear programming relaxation. We use the DTC($G$) to denote the convex hull of all vectors $x$
satisfying the constraints above (with integrity constraints replaced by $x \geq 0$). 
We express below the dual of $DTC($G$)$:
\[ \max \sum_{S \in \mathcal{F}}y_{S} \]
$$\sum_{S \in \mathcal{F} \mbox{ s.t. } e \in \delta^-(S)}y_S \leq c(e) \mbox{ for all $e \in A$}$$
$$y_S \geq 0 \mbox{ for all $S \in \mathcal{F}$}$$
\section{Approximating the minimum $r$-branching cover}
\noindent
\subsection{Preliminary observations and algorithm overview}
\subsubsection{Preliminary observations}
As we can see, the minimum $r$-branching cover is closely related to the well-known minimum $r$-arborescence problem which finds a minimum 
$r$-branching spanning all the nodes in $G$. Edmonds \cite{Edmonds} gave a linear programming formulation for this problem which consists of the cut constraints 
for all the subsets $S \subseteq V \setminus \{r\}$ (not limited to $S \in \mathcal{F}$). He designed then a primal-dual algorithm (also described in \cite{Chu}) which repeatedly keeps and updates 
a set $A_0$ of zero reduced cost and the subgraph $G_0$ induced by $A_0$ and at each iteration, tries to cover a chosen strongly connected component in $G_0$ by augmenting 
(as much as possible with respect to the current reduced cost) the corresponding dual variable. 
The algorithm ends when all the nodes are reachable from $r$ in $G_0$. The crucial point in the Edmonds' algorithm is that when there still exist nodes not reachable from $r$ in $G_0$, 
there always exists in $G_0$ a strongly connected component to be covered because we can choose trivial strongly connected components which are singletons. We can not do such a thing for 
minimum $r$-branching cover because a node can be or not belonging to a $r$-branching cover. But we shall see that if $G_0$ satisfies a certain conditions, we can use an Edmonds-style primal-dual algorithm to find a $r$-branching cover and to obtain a $G_0$ satisfying such conditions, we should pay a ratio of $\max\{2,\ln(n)\}$. Let us see what could be these conditions. A node $j$ is said \textit{connected to} to another node $i$ (resp. a  connected subgraph $B$) if there is a path from $i$ (resp. a node in $B$) to $j$. 
Suppose that we have found a vertex cover $U$ and a graph $G_0$, we define an \textit{Edmonds connected subgraph} as a non-trivial connected (not necessarily strongly) subgraph $B$ not containing $r$ of $G_0$ such that
given any node $i \in B$ and for all $v \in B \cap U$, $v$ is connected to $i$ in $G_0$. Note that any strongly connected subgraph not containing $r$ in $G_0$ which contains at least 
a node in $U$ is an  Edmonds connected subgraph. As in the definition, for an Edmonds connected subgraph $B$, we will also use abusively $B$ to denote its vertex set.
\begin{theorem} \label{th1}
 If for any node $v \in U$ not reachable from $r$ in $G_0$, we have 
\begin{itemize}
 \item either $v$ belongs to an Edmonds connected subgraph of $G_0$, 
\item or $v$ is connected to an Edmonds connected subgraph of $G_0$.
\end{itemize}
then we can apply an Edmonds-style primal-dual algorithm completing $G_0$ to get a $r$-branching cover spanning $U$ without paying any additional ratio.
\end{theorem}
\begin{proof}
We will prove that if there still exist nodes in $U$ not reachable from $r$ in $G_0$, then there always exists an Edmonds connected subgraph, say $B$, 
uncovered, i.e. $\delta_{G_0}^-(B)=\emptyset$. Choosing any node $v_1 \in U$ not reachable from $r$ in $G_0$, we can see that in both cases, the Edmonds connected subgraph, say $B_1$, of $G_0$ containing $v_1$ or to which $v_1$ is connected, is not reachable from $r$. In this sense we suppose that $B_1$ is maximal. If $B_1$ is uncovered, we have done. If $B_1$ is covered then it should be covered by an arc from a node $v_2 \in U$ not reachable from $r$ because if $v_2 \notin U$ then $B_1 \cup \{v_2\}$ induces an Edmonds connected subgraph which contradicts the fact that $B_1$ is maximal. Similarly, we should have $v_2 \neq v_1$ because otherwise $B_1 \cup \{v_1\}$ induces an Edmonds connected subgraph.
We continue this reasoning with $v_2$, if this process does not stop, we will meet another node $v_3 \in U \setminus \{v_1,v_2\}$ not reachable from $r$ and so on \ldots. As $|U| \leq n-1$, this process should end with an  Edmonds connected subgraph $B_k$ uncovered. \\
We can then apply a primal-dual Edmonds-style algorithm (with respect to the reduced cost modified by the determination of $U$ and $G_0$ before) 
which repeatedly cover in each iteration an uncovered Edmonds connected subgraph in $G_0$ until every node in $U$ is reachable 
from $r$. By definition of Edmonds connected subgraphs, in the output $r$-branching cover, we can choose only one arc entering the chosen Edmonds connected subgraph and it is enough to 
cover the nodes belonging to $U$ in this subgraph.  
\end{proof}
\subsubsection{Algorithm overview} 
Based on the above observations on $DTC(G)$ and its dual, we  design an algorithm 
which is a composition of 3 phases. Phases I and II determine $G_0$ and a vertex cover $U$ satisfying the conditions stated in Theorem \ref{th1}. 
The details of each phase is as follows:
\begin{itemize}
\item Phase I is of a primal-dual style which tries to cover the sets $S \in \mathcal{F}$ 
such that $|S|=2$. We keep a set $A_0$ of zero reduced cost and the subgraph $G_0$ induced by $A_0$.
$A_0$ is a cover but does not necessarily contain a $r$-branching cover. We determine after this phase a vertex cover (i.e. a node cover) set $U$ of $G$.
Phase I outputs a partial solution $T^1_0$ which is a directed tree rooted in $r$ spanning the nodes in $U$ reachable from $r$ in $G_0$.
It outputs also a dual feasible solution $y$.
 \item Phase II is executed only if $A_0$ does not contain a $r$-branching cover, i.e. there are nodes in $U$ determined in Phase I 
 which are not reachable from $r$ in $G_0$. Phase II works with the reduced costs issued from Phase I and tries to make the nodes in $U$ not reachable from $r$ in $G_0$, either reachable from $r$ in $G_0$, or belong or be connected to an Edmonds connected subgraph in $G_0$.
Phase II transforms this problem to a kind of Set Cover Problem and solve it by a greedy algorithm. Phase II outputs a set of arcs $T^2_0$ and grows the dual solution $y$ issued from Phase I (by growing only the zero value components of $y$). 
\item Phase III is executed only if $T^1_0 \cup T^2_0$ is not a $r$-branching cover. Phase III applies 
a primal-dual Edmonds-style algorithm (with respect to the reduced cost issued from Phases I and II) 
which repeatedly cover in each iteration an uncovered Edmonds connected subgraph in $G_0$ until every node in $U$ is reachable from $r$. 
\end{itemize}
\noindent
\subsection{Initialization}
Set $\mathcal{B}$ to be the collection of the vertex set of all the arcs in $A$ which do not have $r$ as an end vertex. 
In other words, $\mathcal{B}$ contains all the sets of cardinality 2 in $\mathcal{F}$, i.e. $\mathcal{B} = \{S ~ | \mbox{$S \in \mathcal{F}$ and $|S|=2$}\}$. 
Set the dual variable to zero, i.e. $y\leftarrow 0$ and set the reduced cost $\bar{c}$ to $c$, i.e. $\bar{c}\leftarrow c$.
Set $A_0 \leftarrow \{e \in A ~|~ \bar{c}(e)=0\}$. Let $G_0=(V_0,A_0)$ be the subgraph of $G$ induced by $A_0$.\\
During the algorithm, we will keep and update constantly a subset of $T_0 \subseteq A_0$. 
At this stage of initialization, we set $T_0 \leftarrow \emptyset$.\\
 During Phase I, we also keep updating a dual feasible solution $y$ that is initialized at 0 (i.e. all the components of $y$ are equal to 0). The dual solution $y$ is 
not necessary in the construction of a $r$-branching cover but we will need it in the proof for the performance guarantee of the algorithm. 
\subsection{Phase I}
In this phase, we will progressively expand $A_0$ so that it covers all the sets in $\mathcal{B}$. In the mean time, 
during the expansion of $A_0$, we add the vertex set of  newly created strongly connected components of $G_0$ to $\mathcal{B}$.\\
Phase I repeatedly do the followings until $\mathcal{B}$ becomes empty.
\begin{enumerate}
 \item select a set $S \in \mathcal{B}$ which is not covered by $A_0$.
\item select the cheapest (reduced cost) arc(s) in $\delta^-(S)$  
 and add it (them) to $A_0$. $A_0$ covers then $S$.
Let $\alpha$ denote the reduced cost of the cheapest arc(s) chosen above, then we modify the reduced cost of the arcs in $\delta^-(S)$ by subtracting $\alpha$ from them.
Set $y_S \leftarrow \alpha$.
\item  Remove $S$ from $\mathcal{B}$ and if we detect a strongly connected component $K$ in $G_0$ due to the addition of new 
arcs in $A_0$, in the original graph $G$, we add the set $V(K)$ to $\mathcal{B}$.
\end{enumerate} 
\begin{proposition}\label{re:cover}
 After Phase I, $A_0$ is a cover.
\end{proposition}
\begin{proof}
As we can see, Phase I terminates when $\mathcal{B}$ becomes empty. That means the node sets of the arcs, 
which do not have $r$ as an end-node, are all covered by $A_0$. 
Also all the strongly connected components in $G_0$ are covered. $\hfill$ $\Box$
\end{proof}
At this stage, if for any node $v$ there is a path from $r$ to $v$ in $G_0$, we say that $v$ is \textit{reachable} from 
$r$. Set $T_0$ to be a directed tree (rooted in $r$) in $G_0$ spanning the nodes reachable from $r$. $T_0$ is chosen such that for each strongly connected component $K$ added to $\mathcal{B}$ in Phase I, there is exactly one arc in $T_0$ entering $K$, i.e. $|\delta^-(K) \cap T_0|=1$.    
If the nodes reachable from $r$ in $G_0$ form 
a vertex cover, then $T_0$ is a $r$-branching cover and the algorithm stops.  Otherwise, it goes to Phase II.\\
\subsection{Phase II}
Let us consider the nodes which are not reachable from $r$ in $G_0$. 
We divide them into three following categories: 
\begin{itemize}
 \item The nodes $i$ such that $|\delta^-_{G_0}(i)|=0$, i.e. there is no arc in $A_0$ entering $i$. Let us call these nodes \textit{source nodes}.
\item The nodes $i$  such that $|\delta^-_{G_0}(i)|=1$, i.e. there is exactly one arc in $A_0$ entering $i$. Let us call these nodes \textit{sink nodes}.
\item The nodes $i$  such that $|\delta^-_{G_0}(i)|\geq 2$, i.e. there is at least two arcs in $A_0$ entering $i$. Let us call these nodes \textit{critical nodes}.
\end{itemize}
\begin{proposition}
 The set of the source nodes is a stable set.
\end{proposition}
\begin{proof}
Suppose that the converse is true, then there is an arc $(i,j)$ with $i$, $j$ are both source nodes. As $\delta^-_{G_0}(i)= \delta^-_{G_0}(j)=\emptyset$, we have 
$\delta^-_{G_0}(\{i,j\})=\emptyset$. Hence, $(i,j)$ is not covered by $A_0$. Contradiction.
\end{proof}
\begin{corollary}
 The set $U$ containing the nodes reachable from $r$ in $G_0$ after Phase I, the sink nodes and the critical nodes is a vertex cover (i.e. a node cover) of $G$.
\end{corollary}
\begin{proposition}\label{pro:sink}
 For any sink node $j$, there is at least one critical node $i$ such that $j$ is connected to $i$ in $G_0$. 
 \end{proposition}
\begin{proof}
 Let the unique arc in $\delta^-_{G_0}(j)$ be $(i_1,j)$. Since this arc should be covered by $A_0$, $\delta^-_{G_0}(i_1) \neq \emptyset$. 
If  $|\delta^-_{G_0}(i_1)| \geq 2$ then $i_1$ is a critical node and we have done. Otherwise, i.e. $|\delta^-_{G_0}(i_1)|=1$ and 
$i_1$ is a sink node. Let $(i_2,i_1)$ be the unique arc in $\delta^-_{G_0}(i_1)$, we repeat then the same reasoning for $(i_2,i_1)$ and 
for $i_2$. If this process does not end with a critical node, it should meet each time a new sink node not visited before (It is not possible that a 
directed cycle is created since then this directed cycle (strongly connected component) should be covered in Phase I and hence at least 
one of the nodes on the cycle has two arcs entering it, and is therefore critical). As the number of sink nodes is at most $n-1$,  the process can not continue infinitely and should end at a stage $k$ ($k<n$) with $i_k$ is a critical node. By construction, the path $i_k,i_{k-1},\ldots,i_1,j$ is a path in $G_0$ from $i_k$ to $j$.
\end{proof}
A critical node $v$ is said to be \textit{covered} if there is at least one arc $(w,v) \in A_0$ such that $w$ is not a source node, i.e. $w$ can be  
a sink node or a critical node or a node reachable from $r$. Otherwise, we say $v$ is \textit{uncovered}.
\begin{proposition}\label{pro:critical}
 If all critical nodes are covered then for any critical node $v$, one of the followings is verified:
\begin{itemize}
 \item either $v$ belongs to an Edmonds connected subgraph of $G_0$ or $v$ is connected to an Edmonds connected subgraph of $G_0$,
\item  there is a path from $r$ to $v$ in $G_0$, i.e. $v$ is reachable from $r$ in $G_0$.
\end{itemize}
\end{proposition}
\begin{proof}
If $v$ is covered by a node reachable from $r$, we have done. Otherwise, $v$ is covered by sink node or by another critical node. From Proposition 
\ref{pro:sink} we derive that in the both cases, 
$v$ will be connected to a critical node $w$, i.e. there is a path from $w$ to $v$ in $G_0$. Continue this reasoning with $w$ and so on, we should end with a node reachable from $r$ or a critical node visited before. In the first case $v$ is reachable from $r$. In the second case, $v$ belongs to a directed cycle in $G_0$ if we have revisited $v$, otherwise $v$ is connected to a directed cycle in $G_0$. The directed cycle in the both cases is an Edmonds connected subgraph (because it is strongly connected) and it can be included in a greater Edmonds connected subgraph.
\end{proof}
\begin{lemma}\label{lem:critical}
 If all critical nodes are covered then for any node $v \in U$ not reachable from $r$ in $G_0$,
\begin{itemize}
 \item either $v$ belongs to an Edmonds connected subgraph of $G_0$,
\item or $v$ is connected to an Edmonds connected subgraph of $G_0$.
\end{itemize}
\end{lemma}
\begin{proof}
 The lemma is a direct consequence of Propositions \ref{pro:sink} and \ref{pro:critical}.
\end{proof}
The aim of Phase II is to cover all the uncovered critical nodes. Let us see how to convert this problem into a weighted SCP and to solve the latter by adapting the well-known greedy algorithm for weighted SCP.\\
A source node $s$ is  
\textit{zero connecting} a critical node $v$ (reciprocally $v$ is  \textit{zero connected from} $s$) if $(s,v) \in A_0$. 
If $(s,v) \notin A_0$ but $(s,v) \in A$ then $s$ is 
\textit{positively connecting}  $v$ (reciprocally $v$ is \textit{positively connected from} $s$).\\
Suppose that at the end of Phase I, there are $k$  uncovered critical nodes $v_1$, $v_2$, \ldots $v_k$ and $p$ source nodes $s_1$, $s_2$ , \ldots $s_p$. 
Let $S=\{s_1,s_2, \ldots,s_p\}$ denote the set of the source nodes.
\begin{remark}\label{cover_critical}
An uncovered critical node $v$ can be only covered:
\begin{itemize}
\item by directly an arc from a sink node or another crtitical node to $v$,
\item or via a source node $s$ connecting (zero or positively) $v$, i.e. by two arcs: an arc in $\delta^-(s)$ and the arc $(s,v)$. 
\end{itemize}
\end{remark}
Remark \ref{cover_critical} suggests us that we can consider every critical node $v$ as a ground element to be covered in a Set Cover instance and 
the subsets containing $v$ could be the singleton $\{v\}$ and any subset containing $v$ of the set of the critical nodes connecting (positive or zero) from $s$.    
The cost of the the singleton $\{v\}$ is the minimum reduced cost of the arcs from a sink node or another crtitical node to $v$. 
The cost of a subset $T$ containing $v$ of the set of the  critical nodes connecting from $s$ is the minimum reduced cost of the arcs in $\delta^-(s)$ plus 
the sum of the reduced cost of the arcs $(s,w)$ for all $w \in T$.\\ 
Precisely, in Phase II, we proceed to cover all the uncovered critical nodes by solving by the greedy algorithm the following instance of the Set Cover Problem:
\begin{itemize}
 \item The ground set contains $k$ elements which are the critical nodes $v_1$, $v_2$, \ldots, $v_k$.
\item The subsets are
\begin{description}
 \item[Type I]  For each source node $s_i$ for $i=1,\ldots,p$, let $\mathcal{C}(s_i)$ be the set of all the critical nodes 
connected (positively or zero) from $s_i$. The subsets of Type I associated to $s_i$ are the subsets of $\mathcal{C}(s_i)$ ($\mathcal{C}(s_i)$ 
included). To define their cost, we define 
$$\bar{c}(s_i)= \left\{ \begin{array}{ll}
\min\{\bar{c}(e) ~| ~ \mbox{$e \in \delta^-(s_i)$} \} & \mbox{ if $\delta^-(s_i) \neq \emptyset$,}\\
+\infty & \mbox{ otherwise}
\end{array}
\right.
$$
Let us choose an arc $e_{s_i}= \mathrm{argmin}\{\bar{c}(e) ~| ~ \mbox{$e \in \delta^-(s_i)$} \}$ which denotes an arc entering $s_i$ of minimum reduced cost. Let $T$ be any subset of type I associated to $s_i$, we define $\bar{c}(T)$ the cost of $T$ 
as $\displaystyle \bar{c}(T)=\bar{c}(s_i)+\sum_{v \in T}\bar{c}(s_i,v)$.   
Let us call the arc subset containing the arc $e_{s_i}$ and the arcs $(s_i,v)$ for all $v \in T$ uncovered, \textit{the covering arc subset of $T$}.
 \item[Type II] the singletons $\{v_1\}$, $\{v_2\}$, \ldots, $\{v_k\}$. We define the cost of the singleton $\{v_i\}$, 
 $$\bar{c}(v_i) = 
\left\{ 
\begin{array}{ll}
\min\{\bar{c}(w,v_i) ~|~ \mbox{where $w$ is not a source node, i.e. $w \in V \setminus S$}\} & \mbox{ if $(V \setminus S : \{v_i\}) \neq \emptyset$, }\\
+\infty & \mbox{otherwise}
\end{array}
\right.$$
\end{description}
Let us choose an arc $e_{v_i}=\mathrm{argmin}\{\bar{c}(w,v_i) ~|~ \mbox{where $w$ is not a source node, i.e. $w \in V \setminus S$}\}$, denotes an arc entering $v_i$ from a non source node of minimum reduced cost. Let the singleton $\{e_{v_i}\}$ be \textit{the covering arc subset of $\{v_i\}$}.
\end{itemize}
\noindent
We will show that we can adapt the greedy algorithm solving this set cover problem to our primal-dual scheme. In particular, we will specify how to update dual variables et the sets 
$A_0$ and $T_0$ in each iteration of the greedy algorithm.
The sketch of the algorithm is explained in Algorithm 1.
\begin{algorithm}\label{greedy}
\While{ there exist uncovered critical nodes}{
Compute the most efficient subset $\Delta$ \;
Update the dual variables and the sets $A_0$ and $T_0$\;
Change the status of the uncovered critical nodes in $\Delta$ to covered \;
}
\caption{Greedy algorithm for Phase II}
\end{algorithm}
\noindent
Note that in Phase II, contrary to Phase I, the reduced costs $\bar{c}$ are not to be modified and all the computations are based on 
the reduced costs $\bar{c}$ issued from Phase I. In the sequel, we will specify how to compute the most efficient subset $\Delta$ 
and update the dual variables.\\
For $1 \leq i \leq p$ let us call $\mathcal{S}_i$  the collection of all the subsets 
of type I associated to $s_i$. 
Let $\mathcal{S}$ be the collection of all the subsets of type I and II.\\
\textbf{Computing the most efficient subset. }
Given a source node $s_i$, while the number of subsets in $\mathcal{S}_i$ can be exponential, we will show 
in the following that computing the most efficient subset in $\mathcal{S}_i$ is can be done in polynomial time. 
Let us suppose that there are $i_q$ critical nodes denoted by $v^{i_1}_{s_i}, v^{i_2}_{s_i}, \ldots, v^{i_q}_{s_i}$ which are connected (positively or zero) from $s_i$. 
In addition, we suppose without loss of generality that $\bar{c}(s_i,v^{i_1}_{s_i}) \leq \bar{c}(s_i,v^{i_2}_{s_i}) \leq \ldots \leq \bar{c}(s_i,v^{i_q}_{s_i})$.
We  compute $f_i$ and $S_i$ which denote respectively the best efficiency and the most effecicient set  
in $\mathcal{S}_i$ by the following algorithm.
\begin{description}
 \item [Step 1 ] Suppose that $v^{i_h}_{s_i}$ is the first uncovered critical node met when we scan the critical nodes $v^{i_1}_{s_i}, v^{i_2}_{s_i}, \ldots, v^{i_q}_{s_i}$ in this order.\\    
Set $S_{i} \leftarrow \{v^{i_h}_{s_i}\}$. Set $\bar{c}(S_i) \leftarrow \bar{c}(s_i)+\bar{c}(s_i,v^{i_h}_{s_i})$.\\ 
Set $d_i \leftarrow 1$. Set $f_i \leftarrow \frac{\bar{c}(S_i)}{d_i}$ and $\Delta_i \leftarrow S_i$.
\item [Step 2] We add progressively uncovered critical nodes $v^{i_j}_{s_i}$ for $j=h+1,\ldots, i_q$ to $S_i$ while 
this allows to increase the efficiency of $S_i$:\\
For $j=h+1$ to $i_q$, if $v^{i_j}_{s_i}$ is uncovered and $f_{i}>\frac{\bar{c}(S_i)+\bar{c}(s_i,v^{i_j})}{d_i+1}$ then $f_{i} \leftarrow \frac{\bar{c}(S_i)+\bar{c}(s_i,v^{i_j}_{s_i}}{d_i+1}$, 
$d_i \leftarrow d_i+1$ and $S_i \leftarrow S_i \cup \{v^{i_j}_{s_i}\}$.
\end{description}
Set $i_{min} \leftarrow \mathrm{argmin}\{f_i ~ | \mbox{$s_i$ is a source node}\}$. \\
Choose the most efficient subset among $S_{i_{min}}$ and the singletons of type II for which the computation of efficiency is 
straightforward. 
Set $\Delta$ to be most efficient subset and set  $d \leftarrow |\Delta|$ the number of the uncovered critical nodes in $\Delta$.\\
\textbf{Updating the dual variables and the sets $A_0$ and $T_0$}\\
Let $g = \max\{|T| ~ | ~ T \in \mathcal{S}\}$ and let $H_g=1+\frac{1}{2}+\frac{1}{3}+\ldots+\frac{1}{g}$.  
\begin{remark}\label{g}
 $g \leq D^+_r$.
\end{remark}
Given a critical node $v$, let $p_v$ denote the number of source nodes connecting $v$. 
Let $s^v_1$, $s^v_2$, \ldots, $s^v_{p_v}$ be these source nodes such that $\bar{c}(s^v_1,v) \leq \bar{c}(s^v_2,v)\leq \ldots \leq \bar{c}(s^v_{p_v},v)$. 
We define $S^j_v=  \{v,s^1_v, \ldots, s^j_v\}$ for $j=1, \ldots, p_v$.  
We can see that for $j=1, \ldots,p_v$, $S^j_v \in \mathcal{F}$. Let $y_{S^j_v}$ be the dual variable associated to the cut constraints 
$x(\delta^-(S^j_v)) \geq 1$. 
The dual variables will be updated as follows.
For each critical node $v$ uncovered in $\Delta$, we update the value of $y_{S^j_v}$ for $j=1, \ldots, p_v$ 
for that $\sum_{j=1}^{p_v}y_{S^j_v}=\frac{\bar{c}(\Delta)}{H_g\times d}$. This updating process 
saturates progressively the arcs $(s^j_v,v)$ for $j=1, \ldots, p_v$. Details are given in Algorithm \ref{update}.
\begin{algorithm}  \label{update}
$j \leftarrow 1$ \;
\While{ ($j<p_v$) and ($\bar{c}(s^{j+1}_v,v) < \frac{\bar{c}(\Delta)}{H_g\times d}$)}{
$y_{S^j_v} \leftarrow \bar{c}(s^{j+1}_v,v)-\bar{c}(s^{j}_v,v)$\;
$j \leftarrow j+1$ \;
}
\If{$\bar{c}(s^{p_v}_v,v)<\frac{\bar{c}(\Delta)}{H_g\times d}$}{
	 $y_{S^{p_v}_v} \leftarrow \frac{\bar{c}(\Delta)}{H_g\times d} - \bar{c}(s^{p_v}_v,v)$\;
}
\caption{Updating the dual variables}
\end{algorithm}
We add to $A_0$ and to $T_0$ the arcs in the covering arc subset of $\Delta$. \\
Let us define $\mathcal{T}$ as the set of the subsets $T$ such that $y_T$ is made positive in Phase II.
\begin{lemma}\label{phase2_valid}
 The dual variables which were made positive in Phase II respect the reduced cost issued from Phase I.
\end{lemma}
\begin{proof}
For every $T \in \mathcal{T}$, the arcs in $\delta^-(T)$ can only be 
either an arc in $\delta^-(s_i)$ with $s_i$ is a source node or an arc in $\delta^-(v)$ with 
$v$ is a critical node. Hence, we should show that for every arc $(u',u)$ with $u$ is either a critical node or a source node, we have 
$$\displaystyle \sum_{T \in \mathcal{T} \mbox{ s.t. } u \in T}y_T \leq \bar{c}(u',u)$$ 
\begin{itemize}
 \item $u$ is a critical node $v$ and $u'$ is the source node $s^v_j$. The possible subsets $T \in \mathcal{T}$ such that $(s^v_j,v) \in \delta^-(T)$ are the 
sets $S^1_v$, \ldots, $S^{j-1}_v$. By Algorithm \ref{update}, we can see that 
$$\displaystyle \sum_{k=1}^{j-1} y_{S_v^{k}} \leq \bar{c}(s^v_j, v).$$
\item $u$ is a critical node $v$ and $u' \in V \setminus S$. 
By definition of $\bar{c}(v)$, we have $\bar{c}(u',u) \geq \bar{c}(v)$. By analogy with the Set Cover problem, 
the dual variables made positive in Phase II respect the cost of the singleton $\{v\}$. Hence 
$$\displaystyle \sum_{T \in \mathcal{T} \mbox{ s.t. } v \in T}y_T \leq \bar{c}(v) \leq \bar{c}(u',u)$$ 
\item $u$ is source node and $u' \in V \setminus S$. 
For each critical node $w$ such that $(u,w) \in A$, we suppose that 
$u=s^{i(u,w)}_w$ where $1 \leq i(u,w) \leq p_w$. Let 
$$T_u=\{ w ~ | ~ \mbox{$w$ is a critical node,  $(u,w) \in A$ and $y_{S^{i(u,w)}_w}>0$} \}$$
We can see that $T_u \in \mathcal{S}$ and $\bar{c}(T_u) = \bar{c}(u) + \sum_{w \in T_u}\bar{c}(u,w)$.
Suppose that $l$ is the total number of iterations in Phase II. 
We should show that
\begin{equation}\label{dual1}
 \sum_{k=1}^l\sum_{w \in T_u \cap \Delta_k}(\frac{\bar{c}(\Delta_k)}{H_g \times d_k}-\bar{c}(u,w)) \leq \bar{c}(u)
\end{equation}
where $\Delta_k$ is the subset which has been chosen in $k^{th}$ iteration.
Let $a_k$ be the number of uncovered critical nodes in $T_u$ at the beginning of the $k^{th}$ iteration. We have then $a_1=|T_u|$ and 
$a_{l+1}=0$. 
Let $A_k$ be the set of previously uncovered critical nodes of $T_u$ covered in the $k^{th}$ iteration. We immediately find that 
$|A_k|=a_k -a_{k+1}$. By Algorithm \ref{greedy}, we can see that at the $k^{th}$ iteration
$\frac{\bar{c}(\Delta_k)}{H_g \times d_k} \leq \frac{\bar{c}(T_u)}{H_g \times a_k}$. Since $|A_k|=a_k -a_{k+1}$ then 
$$\sum_{w \in T_u \cap \Delta_k}(\frac{\bar{c}(\Delta_k)}{H_g \times d_k})-\sum_{w \in T_u \cap \Delta_k}\bar{c}(u,w)) \leq \frac{\bar{c}(T_u)}{H_g} \times \frac{a_k-a_{k+1}}{a_k}-\sum_{w \in T_u \cap \Delta_k}\bar{c}(u,w)) $$
Hence, 
\begin{eqnarray*}
 \sum_{k=1}^l\sum_{w \in T_u \cap \Delta_k}(\frac{\bar{c}(\Delta_k)}{H_g \times d_k}-\bar{c}(u,w))  & \leq & \frac{\bar{c}(T_u)}{H_g}\sum_{k=1}^l \frac{a_k-a_{k+1}}{a_k} - \sum_{k=1}^l\sum_{w \in T_u \cap \Delta_k}\bar{c}(u,w)) \\
 & \leq & \frac{\bar{c}(T_u)}{H_g}\sum_{k=1}^l(\frac{1}{a_k} + \frac{1}{a_{k}-1} + \ldots + \frac{1}{a_{k+1}-1}) - \sum_{k=1}^l\sum_{w \in T_u \cap \Delta_k}\bar{c}(u,w))\\
& \leq & \frac{\bar{c}(T_u)}{H_g}\sum_{i=1}^{a_1}\frac{1}{i} - \sum_{k=1}^l\sum_{w \in T_u \cap \Delta_k}\bar{c}(u,w)) \\
&\leq  & \bar{c}(T_u)- \sum_{k=1}^l\sum_{w \in T_u \cap \Delta_k}\bar{c}(u,w)) = \bar{c}(u).
\end{eqnarray*}
 \end{itemize}
\end{proof}
Let $T^2_0 \subset T_0$ the set of the arcs added to $T_0$ in Phase II. For each $e \in T^2_0$, 
let $c_2(e)$ be the part of the cost $c(e)$ used in Phase II. 
\begin{theorem}\label{perf2}
 $$\displaystyle c_2(T_0)=\sum_{e \in T^2_0}c_2(e) \leq H_g \sum_{T \in \mathcal{T}}y_T \leq \ln(D^+_r)\sum_{T \in \mathcal{T}}y_T$$
\end{theorem}
\begin{proof}
By Algorithm \ref{update}, at the $k^{th}$ iteration, a subset $\Delta_k$ is chosen and 
we add the arcs in the covering arc subset of $\Delta_k$ to $T^2_0$ for all 
$v \in \Delta_k$. Let $T^{2_k}_0$ be covering arc subset of $\Delta_k$. We can see that 
$c_2(T^{2_k}_0)=\sum_{e \in T^{2_k}_0}\bar{c}_e=\bar{c}(\Delta_k)$. 
In this iteration, we update the dual variables in such a way 
that for each critical node $v \in \Delta_k$, $\sum_{j=1}^{p_v}y_{S^j_v}=\frac{\bar{c}(\Delta_k)}{H_g \times d_k}$ with $d_k=|\Delta_k|$. 
Together with the fact that $\displaystyle \bar{c}(\Delta_k)=\bar{c}(w_k)+ \sum_{v \in \Delta_k}\bar{c}(w_k,v)$ we have
$\sum_{v \in \Delta_k}\sum_{j=1}^{p_v}y_{S^j_v}=\frac{\bar{c}(\Delta_k)}{H_g}=\frac{c_2(T^{2_k}_0)}{H_g}$.
By summing over $l$ be the number of iterations in Phase II, we obtain 
$$\sum_{T \in \mathcal{T}}y_T=\sum_{k=1}^l\sum_{v \in \Delta_k}\sum_{j=1}^{p_v}y_{S^j_v}=\sum_{k=1}^l\frac{\bar{c}(\Delta_k)}{H_g}=\sum_{k=1}^l\frac{c_2(T^{2_k}_0)}{H_g}=\frac{c_2(T_0)}{H_g}$$
which proves that $c_2(T_0) = H_g \sum_{T \in \mathcal{T}}y_T$. By Remark \ref{g}, we have $g \leq D^+_r$ and $H_g \approx \ln g$, hence $c_2(T_0) \leq \ln(D^+_r)\sum_{T \in \mathcal{T}}y_T$.  \hfill $\Box$
\end{proof}
\subsection{Phase III}
We perform Phase III if after Phase II, there exist nodes in $U$ not reachable from $r$ in $G_0$. 
By Lemma \ref{lem:critical},  
they belong or are connected to some Edmonds connected subgraphs of $G_0$. By Theorem \ref{th1}, we can apply an Edmonds-style primal-dual algorithm which tries to cover uncovered Edmonds connected subgraphs of $G_0$ until all nodes in $U$ reachable from $r$. The algorithm 
repeatedly choosing uncovered Edmonds connected subgraph and adding to $A_0$ the cheapest (reduced cost) arc(s) entering it . As the reduced costs have not been modified during Phase II, 
we update first the reduced cost $\bar{c}$ with respect to the dual variables made positive in Phase II. \\
\begin{algorithm}
Update the reduced cost $\bar{c}$ with respect to the dual variables made positive in Phase II\;
\Repeat{every nodes in $U$ reachable from $r$}{
Choose $B$ an uncovered Edmonds connected subgraph \;
Let $y_{B}$ be the associated dual variable to $B$\;
Set $\bar{c}(B) \leftarrow \min\{\bar{c}_e ~|~ e \in \delta^-(B)\}$ ; Set $y_{B} \leftarrow \bar{c}(B)$\;
\ForEach{ $e \in \delta^-(B)$}{
	$\bar{c}_e \leftarrow \bar{c}_e -\bar{c}(B)$\;  
} 
Update $A_0$, $G_0$ and $T_0$ (see below)\;
}
\caption{Algorithm for Phase III}
\end{algorithm}
For updating $A_0$, at each iteration, we add all the saturated arcs belonging to $\delta^-(B)$ to $A_0$.
Among these arcs, we choose only one arc $(u,v)$ with $v \in B$ to add to $T_0$ with a preference for a $u$ 
connected from $r$ in $G_0$. In the other hand, we delete the arc $(x,v)$ with $x \in B$ from $T_0$. We then add to $T_0$ 
an directed tree rooted in $v$ in $G_0$ spanning $B$. If there are sink nodes directly connected to $B$, i.e. 
the path from a critical node $w \in B$ to these nodes contains only sink nodes except $w$. We also add all such paths to $T_0$. 
\begin{lemma}
After Phase III, $T_0$ is a $r$-branching cover. 
\end{lemma}     
\begin{proof}
We can see that after Phase III, for any critical node or a sink node $v$, there is a path containing only the arcs in $T_0$ from 
$r$ to $v$ and there is exactly one arc in $\delta^-(v) \cap T_0$. 
\end{proof}
\subsection{Performance guarantee} 
We state now a theorem about the performance guarantee of the algorithm.
\begin{theorem}
 The cost of $T_0$ is at most $\max\{2,\ln(D^+_r)\}$ times the cost of an optimal $r$-branching cover.
\end{theorem}
\begin{proof}
Suppose that $T^*$ is an optimal $r$-branching cover of $G$ with respect to the cost $c$.
 First, we can see that the solution $y$ built in the algorithm is feasible dual solution. Hence
 $c^Ty \leq c(T^*)$. Let $\mathcal{B}$ be the set of all the subsets $B$ in Phase I and Phase III 
 ($B$ is either a subset of cardinality 2 in $\mathcal{F}$ or a subset such that the induced subgraph 
 is a strongly connected component or an Edmonds connected subgraph in $G_0$ at some stage of the algorithm). Recall that we have defined 
$\mathcal{T}$ as the set of the subsets $T$ such that $y_T$ is made positive in Phase II.
We have then 
 $\displaystyle c^Ty= \sum_{B \in \mathcal{B}}y_B + \sum_{T \in \mathcal{T}}y_T $. For any arc $e$ in $T_0$, 
 let us divide the cost $c(e)$ into two parts: $c_1(e)$ the part saturated by the dual variables $y_B$ with 
 $B \in \mathcal{B}$ and $c_2(e)$ the part saturated by the dual variables $y_T$ with 
 $B \in \mathcal{T}$. Hence $c(T_0)=c_1(T_0)+c_2(T_0)$. By Theorem \ref{perf2}, we have 
 $c_2(T_0) \leq \ln(D^+_r)\sum_{T \in \mathcal{T}}y_T$ (note that the replacing in Phase III of 
 an arc $(x,v)$ by another arc 
 $(u,v)$ with $v \in B_i$ do not change the cost $c_2(T_0)$). Let us consider any set $B \in \mathcal{B}$
 by the algorithm, $B$ is the one of the followings:
\begin{itemize}
 \item $|B|=2$. As $T_0$ is a branching so that for all vertex $v \in V$, we have $|\delta^-(v) \cap T_0| \leq 1$. 
Hence, $|\delta^-(B) \cap T_0| \leq 2$.
\item $B$ is a vertex set of a strongly connected component or an Edmonds connected subgraph in $G_0$. We can see obviously that by the algorithm $|\delta^-(B) \cap T_0|=1$.
\end{itemize}
These observations lead to the conclusion that $c_1(T_0) \leq 2\sum_{B \in \mathcal{B}}y_B$. Hence 
\begin{eqnarray*}
c(T_0)=c_1(T_0)+c_2(T_0)& \leq & 2\sum_{B \in \mathcal{B}}y_B+ \ln(D^+_r)\sum_{T \in \mathcal{T}}y_T\\
&\leq &\max\{2, \ln(D^+_r)\}c^Ty \leq \max\{2, \ln(D^+_r)\}c(T^*).
\end{eqnarray*}
\end{proof}
\begin{corollary}
 We can approximate the DTCP within a $\max\{2, \ln(D^+)\}$ ratio.
\end{corollary}
\section{Final remarks}
The paper has shown that the weighted Set Cover Problem is a special case of the Directed Tree Cover Problem and the latter can be approximated with 
a ratio of $\max\{2,\ln(D^+)\}$ (where $D^+$ is the maximum outgoing degree of the nodes in $G$) by a primal-dual algorithm. Based on known complexity results for weighted Set Cover, in one direction, this approximation seems to be best possible.\\
In our opinion, an interesting question is whether the same techniques can be applied to design a combinatorial approximation algorithm for Directed Tour Cover. As we have seen in 
Introduction section, a $2\log_2(n)$-approximation algorithm for Directed Tour Cover has been given in \cite{Nguyen2}, but this algorithm is not combinatorial. 

\end{document}